\documentclass[onecolumn,prx,superscriptaddress,longbibliography,nofootinbib]{revtex4-2}

\usepackage[dvips]{graphicx} % for figures
\usepackage{amsfonts,amscd,amsmath,amsthm}
\usepackage{enumerate}
\usepackage{epsfig}
\usepackage{subfigure}
\usepackage{xcolor}
\usepackage[colorlinks = true]{hyperref}
\usepackage{physics}
\usepackage{epstopdf}
\usepackage{framed}
\usepackage{multirow}
\usepackage{color}
\usepackage{longtable}
\usepackage{comment}
\usepackage[ruled,vlined]{algorithm2e}
\usepackage[most]{tcolorbox}

\graphicspath{{./figure/}}

\usepackage{tikz}
\usetikzlibrary{tikzmark, calc, fit, positioning}
\usetikzlibrary{shapes}
\usepackage{appendix}

\newtheorem{lemma}{Lemma}

\newtheorem{proposition}{Proposition}

\newtcolorbox[auto counter]{mybox}[2][]{
	enhanced,
	breakable,
	colback=blue!5!white,
	colframe=blue!75!black,
	fonttitle=\bfseries,
	title=Box \thetcbcounter: #2,#1
}

\begin{document}

\title{Spin minimum uncertainty states for refined uncertainty relations}
\author{Hao Dai}
\email{dhao@bimsa.cn}
\affiliation{Beijing Institute of Mathematical Sciences and Applications, Beijing 101408 , China}

\author{Yue Zhang}
\email{Corresponding author, zhangyue115@amss.ac.cn}
\affiliation{State Key Laboratory of Mathematical Sciences, Academy of Mathematics and Systems Science, Chinese Academy of Sciences,  Beijing 100190, China}
\affiliation{School of Mathematical Sciences, University of  Chinese Academy of Sciences, Beijing 100049, China}

\begin{abstract}
Minimum uncertainty states of the conventional Heisenberg uncertainty relation have been extensively studied and are often regarded as the most classical quantum states from the perspective of uncertainty, providing valuable insight into the nature of quantumness and its potential applications. In this work, we investigate the minimum uncertainty states associated with an information-theoretic refinement of the Heisenberg uncertainty relation in general spin systems. Using two different approaches, the matrix formulation and the Wick symbol representation, we derive explicit expressions for the states that saturate the uncertainty bound. We show that spin coherent states indeed achieve minimum uncertainty, consistent with their conventional identification as the classical states of spin systems. Moreover, we also identify additional classes of minimum uncertainty states beyond the coherent family. Finally, we compare the spin-system results with the previously studied bosonic case and elucidate the origin of the differences between the two settings. 
\end{abstract}


\maketitle

\section{Introduction}

Inherent to quantum mechanics and fundamentally distinct from classical physics, the well-known Heisenberg uncertainty relation asserts that certain pairs of physical observables cannot be simultaneously measured with arbitrary precision, which is a cornerstone of our understanding of quantum theory \cite{Heis1927,Robe1929}. Since its inception, a wide range of mathematical frameworks and generalizations have been developed to extend the conventional uncertainty relation to diverse quantum systems \cite{Kenn1927,Dodo1980,Maas1988,Hall2004,luo2005,Luo2006,Wu2009,Wehn2010,Rast2013,Bran2013,Busc2014,Kech2014,Cole2017,Dodo2018,Kett2020,Fade2022,Fais2023,Hall2023,Fan2024}. These ongoing efforts have not only enriched the theoretical foundations of quantum mechanics but have also led to significant applications across various domains of quantum science, including quantum entanglement \cite{Guhn2004,Horo2009}, Einstein-Podolsky-Rosen steering \cite{Reid2009}, quantum coherence \cite{Luo2017,Yuan2017,Zhan2020,Zhao2022,Fu2023,Fu2025}, quantum state estimation \cite{Zhu2022}, quantum computation \cite{Rene2016,Rene2017,Li20241,Li20242,Wang2025}, and quantum metrology \cite{Giov2006,Pezz2018}.
%Inherently distinct from classical systems, the well-known Heisenberg uncertainty relation indicates that certain pairs of physical properties cannot be precisely measured simultaneously, which is a cornerstone of the  understanding of quantum mechanics \cite{Heis1927,Robe1929}. Thereafter, a wide range of mathematical approaches have been explored to further extend the conventional uncertainty relation for different quantum systems \cite{Kenn1927,Dodo1980,Maas1988,Hall2004,luo2005,Luo2006,Wu2009,Wehn2010,Rast2013,Bran2013,Busc2014,Kech2014,Cole2017,Dodo2018,Kett2020,Fade2022,Fais2023,Hall2023,Fan2024}. The ongoing explorations have not only deepened theoretical studies, but also practical applications of various areas of quantum theory, such as quantum entanglement \cite{Guhn2004,Horo2009}, Einstein-Podolsky-Rosen steering \cite{Reid2009}, quantum coherence \cite{Luo2017,Yuan2017,Zhan2020,Zhao2022,Fu2023,Fu2025}, quantum state estimation \cite{Zhu2022}, quantum computation \cite{Rene2016,Rene2017,Li20241,Li20242,Wang2025}, and quantum metrology \cite{Giov2006,Pezz2018}.
States for which the uncertainty inequality is satisfied with equality are referred to as minimum uncertainty states. Such states have been the subject of extensive investigation across a wide range of physical systems \cite{Roma1968,Stol1970,Stol1971,Stol1972,Bert1971,Rusc1976,Arag1976,Delb1977,Cani1977,Bacr1978,Cani1978,Milb1984,Orsz1988,Vacc1990,Berg1991,Forb2003,Pegg2005,Fu2020,Li2024}. For the canonical pair of position and momentum observables, it has been established that all minimum uncertainty states of the standard Heisenberg uncertainty relation in infinite-dimensional systems are unitarily equivalent to coherent states \cite{Stol1970,Stol1971,Stol1972}. More recently, it has been shown that all Gaussian states, including mixed Gaussian states, constitute the complete set of minimum uncertainty states for the information-refined Heisenberg uncertainty relation, which provides a tighter bound than the conventional formulation \cite{luo2005,Fu2020}.

%States that make the equality in the uncertainty inequality achieved are referred to as minimum uncertainty states. These states have been the subject of extensive and intensive research in various systems \cite{Roma1968,Stol1970,Stol1971,Stol1972,Bert1971,Rusc1976,Arag1976,Delb1977,Cani1977,Bacr1978,Cani1978,Milb1984,Orsz1988,Vacc1990,Berg1991,Forb2003,Pegg2005,Fu2020,Li2024}. When considering the canonical pair of position and momentum observables, it has been proven that all minimum uncertainty states unitarily correspond to the coherent states for standard Heisenberg uncertainty relation in infinite dimensional system \cite{Stol1970,Stol1971,Stol1972}. More recently, it has been shown that all the Gaussian states, including mixed Gaussian states, are equivalent to minimum uncertainty states for information-refined Heisenberg uncertainty, a tighter inequality compared to standard Heisenberg uncertainty relation \cite{luo2005,Fu2020}.

Finite-dimensional systems, such as spin or angular-momentum systems, have also been studied extensively in the literature \cite{Rusc1976,Arag1976,Delb1977,Bacr1978,Forb2003,Pegg2005,Li2024}. Nevertheless, a complete characterization of minimum uncertainty states in such systems remains an ongoing area of research. For instance, spin coherent states are known to minimize the standard Heisenberg uncertainty relation for angular momentum components \cite{Rusc1976,Arag1976}, and most existing results focus on pure states. Except in two-dimensional systems, the minimum uncertainty states associated with the information-refined uncertainty relation generally include both certain pure states and certain mixed states \cite{Li2024}. This naturally raises the question of how minimum uncertainty states can be characterized in generic finite-dimensional systems. 
%Finite dimensional systems such as spin or angular momentum systems, also have been extensively studied \cite{Rusc1976,Arag1976,Delb1977,Bacr1978,Forb2003,Pegg2005,Li2024}. However, the construction of all minimum uncertainty states remains an ongoing area of research. For example, the spin coherent states minimize the standard Heisenberg uncertainty relation for angular momentum components \cite{Rusc1976,Arag1976}. Most relevant studies to date have focused on pure states. With the exception of 2-dimensional systems, the minimum quantum uncertainty states for the information-refined uncertainty relation consist of both certain pure states and certain mixed states \cite{Li2024}. The natural question then arises: what about the case of generic finite dimensional systems.

The purpose of this work is twofold: (i) to provide a clear characterization of the minimum uncertainty states that saturate the information-refined uncertainty relation for angular momentum operators in generic finite-dimensional systems, and (ii) to elucidate the connection between these finite-dimensional minimum uncertainty states and Gaussian states in infinite-dimensional systems.
%The purpose of this work is twofold: (a) to give a clear characterization of minimum uncertainty states that saturate the information-refined uncertainty relation for angular momentum operators in generic finite-dimensional systems, (b) to show a connection between these minimum uncertainty states in finite-dimensional systems and Gaussian states in infinite-dimensional systems.

This article is organized as follows. In Sec. \ref{sec:unc}, we briefly review the standard Heisenberg uncertainty relation and its information-refined extension. In Sec. \ref{sec:refunc}, we present the information-refined uncertainty relation for spin systems and derive the conditions satisfied by minimum uncertainty states. In Sec. \ref{sec:mat}, we compute the matrix representations of the minimum uncertainty states and show that they correspond to spin coherent states, squeezed states, and certain diagonal states. In Sec. \ref{sec:wic}, we rederive the minimum uncertainty states using the Wick-symbol formalism and establish a one-to-one correspondence between the solutions obtained via the two approaches. Moreover, we show that, in an appropriate limiting regime, the minimum uncertainty states of finite-dimensional systems converge to particular Gaussian states in infinite-dimensional systems. Finally, we conclude with a discussion in Sec. \ref{sec:sum}.

%This article is organized as follows: In Sec. \ref{sec:unc}, we briefly review the standard Heisenberg uncertainty relation and the information-refined uncertainty relation. In Sec. \ref{sec:refunc}, we give the information-refined uncertainty relation for spin system and derive the conditions that minimum uncertainty states satisfy. In Section \ref{sec:mat}, we calculate the matrix form of minimum uncertainty states and find that are corresponding to spin coherent states, squeezed states, and certain diagonal states.
%In Sec. \ref{sec:wic}, we recalculate minimum uncertainty states via symbol calculation and bridge a one-to-one correspondence of solutions throgh two approches. Additionally, we find that in the limit case, the minimum uncertainty states of finite-dimensional systems tend to pariticular Gaussian states of infinite-dimensional systems. In Sec. \ref{sec:sum}, we conclude with a discussion.

\section{Quantum uncertainty relation}\label{sec:unc}

The standard Heisenberg uncertainty relation
\begin{equation}\label{eq:hei}
    V(\rho,X)V(\rho,Y)\geq \frac{1}{4}|\textup{tr}\rho[X,Y]|^2 
\end{equation} 
holds for any two observables $ X $ and $ Y $, and any quantum state $\rho$ \cite{Heis1927,Robe1929}. Here, $V(\rho,X)=\tr(\rho X^2)-(\tr\rho X)^2$ is the variance and $[X,Y]=XY-YX$ is the commutator. The variance contains the uncertainty of the state $\rho$ with respect to the observable $X$ and it actually can be separated into two parts \cite{luo2005}:
\begin{equation*}
V(\rho,X)=I(\rho,X)+C(\rho, X),
\end{equation*}
where $I(\rho,X)=\tr (\rho X^2)-\tr(\sqrt{\rho}X\sqrt{\rho}X),$ and $C(\rho, X)=\text{tr}(\sqrt{\rho}X_{0}\sqrt{\rho}X_{0})$ with $X_{0}=X-\text{tr}\rho X$.
The first part $I(\rho,X)$ is the Wigner-Yanase skew information \cite{Wign1963}. The other part $C(\rho, X)$ arises from the mixedness of the state and thus is called classical uncertainty. 

With the expression of $ X_{0} $, the skew information can be written as
$$I(\rho,X)=\frac{1}{2} \text{tr}(i[\sqrt{\rho},X_{0}])^{2}.$$
Denote
$$J(\rho,X)=\frac{1}{2} \text{tr}\{\sqrt{\rho},X_{0}\} ^{2}.$$
Here $ \{X,Y\}=XY+YX $ represents the anti-commutator.  And we have the relation $J(\rho,X)=V(\rho,X)+ C(\rho, X) .$

For two arbitrary observables $ X $ and $ Y $, the uncertainty relation \cite{luo2005}
\begin{equation} \label{IJ1}
 I(\rho,X)J(\rho,Y) \geq\frac{1}{4}| {\rm tr} \rho [X,Y]|^{2}
\end{equation}
hold.

Furthermore, define the geometric average of $I(\rho,X)$ and $J(\rho,X)$ as $U(\rho,X)$. It can be verified that
\begin{eqnarray*}
U(\rho,X)&=&\sqrt{I(\rho, X)J(\rho, X)}\\
&=&\sqrt{V^{2}(\rho,X)-C^{2}(\rho, X)}.
\end{eqnarray*}
$ U(\rho,X) $ quantifies the quantum uncertainty of $\rho$ with respect of $X$, and satisfies the inequality constraints
$$ F(\rho,X) \leq U(\rho,X)\leq V(\rho,X),$$
where $ F(\rho,X)  $ is the SLD Fisher information. Besides, the two equalities hold if and only if $\rho$ is pure.

The information-refined uncertainty relation is valid \cite{luo2005}, where we replace variance with quantum uncertainty. Specifically, we have
\begin{equation}\label{mix1}
U(\rho,X)U(\rho,Y)\geq \frac{1}{4}|\textup{tr}\rho[X,Y]|^2.
\end{equation}
The main purpose of the article is to explore the states in a spin system for which equality holds. We refer to these states as minimum uncertainty states, as they exhibit certain classical characteristics. 
%investigate the states such that the equality holds in the spin system. We name such states as the minimum uncertainty states, which reflect classical features to some extent.

\section{Spin minimum uncertainty states for refined uncertainty relations}\label{sec:refunc}

In this section, we solve the minimum uncertainty states such that equality holds in Eq.\eqref{mix1} for spin systems. First, we clarify some basic notations in spin systems.

In a spin- $j$ (an integer or half-integer) system with Hilbert space $\mathbb{C}^{2j+1}$, the angular momentum observables satisfy the commutation relation
\begin{equation}
    [J_x,J_y]=iJ_z,\quad[J_y,J_z]=iJ_x,\quad[J_z,J_x]=iJ_y,
\end{equation}
and the ladder operators are denoted by $J_{\pm}=J_x\pm iJ_y$. The $2j+1$ Dicke states $\{\ket{j,m},m=-j,\cdots,j-1,j\}$ constitute an orthonormal basis of the spin system, and satisfy
\begin{align*}
&J_+\ket{j,j}=0,\\
&J_+\ket{j,m} =\sqrt{(j-m)(j+m+1)}\ket{j,m+1}, \qquad m\leq j-1,\\
&J_-\ket{j,-j} =0,\\
&J_-\ket{j,m}  =\sqrt{(j+m)(j-m+1)}\ket{j,m-1}, \qquad m\geq -j+1.
\end{align*}
Denote the observable vector $\boldsymbol{J}=(J_x,J_y,J_z)$ and for a unit vector $\boldsymbol{w}=(w_x,w_y,w_z)$, the rotation operator can be defined as

\begin{equation}
    R=e^{-i\theta \boldsymbol{w}\cdot \boldsymbol{J}},
\end{equation}
where $\theta$ is the rotation angle and $\cdot$ is the inner product of two vectors. The rotation operators are closely related to the displacement operator
\begin{equation}
    \Omega_{\tau}=e^{\tau J_+-\bar{\tau} J_-}
\end{equation}
via $R=\Omega_{\tau}e^{-i\alpha J_{z}}$ with some properly chosen parameters $\tau$ and $\alpha$. Here $\bar{\tau}$ is the complex conjugate.

The spin coherent states, which are regarded as classical states in the spin system, are defined as
\begin{equation}
    \ket{\zeta}=\Omega_{\tau}\ket{j,-j}=(1+|\zeta|^2)^{-j}\sum_{m=-j}^j\sqrt{2j\choose {j+m}}\zeta^{j+m}\ket{j,m},
\end{equation}
where $\tau=\frac\theta2 e^{i\phi}$ and $\zeta={\rm tan}\frac\theta2e^{i\phi}, \ \theta\in[0, \pi], \ \phi\in[0, 2\pi]$ \cite{Radc1971} .
%It should be noticed that when disregarding a non-significant phase, two arbitrary coherent states can be transformed into each other by a displacement operator (or a rotation operator). The highest weight state $\ket{j,j}$ is also a coherent state. Hence, we can also equivalently defined the coherent states as $\Omega_{\zeta}\ket{j,j}$.

It should be noticed that when disregarding a non-significant phase, two arbitrary coherent states can be interconverted via a displacement (or rotation) operator. The lowest weight state $\ket{j,-j}$ is also classified as a coherent state, which allows us to equivalently define the coherent states as $\Omega_{\tau}\ket{j,-j}$.

Now, we turn our attention back to the uncertainty relation given by Eq.\eqref{mix1} for spin systems. Let $ \boldsymbol{n}_{1},\boldsymbol{n}_{2},\boldsymbol{n}_{3} $ form an orthonormal basis for the spin-$j$ system. In Eq.\eqref{mix1}, we can set $ X=J_{\boldsymbol{n}_{1}} $ and $ Y=J_{\boldsymbol{n}_{2} } $, leading to the inequality
%Now we return to consider the uncertain relation Eq.\eqref{mix1} in spin system. Suppose that $ \boldsymbol{n}_{1},\boldsymbol{n}_{2},\boldsymbol{n}_{3} $ forms an orthonormal basis of the spin-$j$ system. In Eq.\eqref{mix1}, let $ X=J_{\boldsymbol{n}_{1}} $ and $ Y=J_{\boldsymbol{n}_{2} } $, then
\begin{equation}\label{mix2} 
U(\rho,J_{\boldsymbol{n}_{1}})U(\rho,J_{\boldsymbol{n}_{2}})\geq \frac{1}{4}|\textup{tr}\rho J_{\boldsymbol{n}_{3}}|^2.
\end{equation}

There exists a rotation operator $R$ such that the three unit vectors, $\boldsymbol{n}_{1}$, $\boldsymbol{n}_{2}$, $\boldsymbol{n}_{3}$ can be rotated to the direction vectors of coordinate $x-$, $y-$, and $z-$axes, respectively. According to the unitary invariance of the quantity $U$, Eq. \eqref{mix2} can be written as
\begin{equation*}
U(R\rho R^{\dagger},J_{x})U(R\rho R^{\dagger},J_{y})\geq \frac{1}{4}|\textup{tr}R\rho R^{\dagger}J_{z}|^2. 
\end{equation*}
Consequently, if we find the solution for the minimum uncertainty states in this specific case,
\begin{equation}\label{mix3} 
U(\rho,J_x)U(\rho,J_y)\geq \frac{1}{4}|\textup{tr}\rho J_z|^2 ,
\end{equation}
the general form of minimum uncertainty states  for Eq. \eqref{mix2} can then be derived by $R^{\dagger}\rho R $ .
%\cite{puri19961}.

From the proof of Eq. \eqref{mix3}, the equality holds if and only if the  two conditions are satisfied simultaneously:

(a) There exists some constant $s\in \mathbb{R}$ such that 
\begin{equation}\label{eq1}
is[\sqrt{\rho} ,J_x]+\{ \sqrt{\rho}, J_{y} - \textup{tr}J_y \rho\} =0
\end{equation}
and actually $s$ satisfies $s^2I(\rho,J_x)=J(\rho,J_y-\textup{tr}J_y \rho).$

(b) There exists some constant $t\in \mathbb{R}$ such that
\begin{equation}\label{eq2}
it[\sqrt{\rho} ,J_y]+\{ \sqrt{\rho}, J_{x} - \textup{tr}J_x \rho\} =0
\end{equation}
and actually $t$ satisfies $t^2 I(\rho,J_y)=J(\rho,J_x-\textup{tr}J_x \rho).$

 Simplify the two above Eqs. \eqref{eq1} and \eqref{eq2} as
\begin{eqnarray} 
%\begin{split}
(s-1)(J_{-}\sqrt{\rho}-\sqrt{\rho}J_{+})+(s+1)(J_{+}\sqrt{\rho}-\sqrt{\rho}J_{-})-iu\sqrt{\rho}&=&0,\label{m3}\\
(t+1)(J_{-}\sqrt{\rho}+\sqrt{\rho}J_{+})+(1-t)(J_{+}\sqrt{\rho}+\sqrt{\rho}J_{-})-v\sqrt{\rho}&=&0,\label{m4}
%\end{split}
\end{eqnarray}
Here, $u=4\tr (\rho J_y), v=4\tr(\rho J_x).$

%As a special case, if $t=1$, the Eq. \eqref{m4} becomes
%\begin{equation}\label{eq:redu}
 %   J_{-}\sqrt{\rho}+\sqrt{\rho}J_{+}=2 \tr(\rho J_x)\sqrt{\rho}.
%\end{equation}
%We claim the state satisfying the Eq. \eqref{eq:redu} is a pure state. Concretely, write the state in the form of spectral decomposition $\rho=\sum_i\lambda_i\ketbra{\psi_i}{\psi_i}$. Suppose
%\begin{equation*}
    %\mu_0=\max_{i}\bra{\psi_i}J_x\ket{\psi_i}=\bra{\psi_0}J_x\ket{\psi_0},
%\end{equation*}
%and calculate the inner product of both sides of Eq. \eqref{eq:redu} with $\ketbra{\psi_0}{\psi_0}$, and
%\begin{equation*}
%\begin{aligned}
 %  \sqrt{\lambda_0}\mu_0 &=\sqrt{\lambda_0}\sum_i\lambda_i\bra{\psi_i}J_x\ket{\psi_i}\\
  % &\leq \sqrt{\lambda_0}\mu_0
  % \end{aligned}
%\end{equation*}

%Otherwise, $\bra{\psi_i}J_x\ket{\psi_i}$ is less than $0$ for all $i$.

\section{The matrix form of minimum uncertainty states}\label{sec:mat}

Suppose the square root of the state can be expressed as $$\sqrt{\rho}=\sum_{m,n} p_{m,n} \ketbra{j,n}{j,m}.$$ 
Denote $c_m=\sqrt{(j-m)(j+m+1)},\quad
-j\leq m\leq j,$ then Eqs. \eqref{m3} and \eqref{m4} equal to
\begin{eqnarray}% \label{m14}
%\begin{split}
(s-1)(c_m p_{m+1,n}-c_n p_{m,n+1})+(s+1)(c_{m-1} p_{m-1,n}-c_{n-1} p_{m,n-1})-iu p_{m,n}&=&0,\label{m14}\\
(t+1)(c_m p_{m+1,n}+c_n p_{m,n+1})+(1-t)(c_{m-1} p_{m-1,n}+c_{n-1} p_{m,n-1})-v p_{m,n}&=&0.\label{m15}
%\end{split}
\end{eqnarray}
To be clear, we solve the linear equations ~\eqref{m3} and \eqref{m4} for spin-$1/2$ and spin-$j$ systems separately.

Now let's assume that $s\neq -1$ and $t\neq 1$ so that the denominator is nonzero in the following calculation. The cases $ s=-1$ or $t=1$ can be treated analogously: identical arguments apply after replacing $1-t$ and $s+1$ in the denominator by $t+1$ and $s-1$, respectively. We return to these special cases at the end of this section.

\subsection{Spin-$1/2 $ system}

When $j=1 /2$, the boundary condition gives that $c_{\frac1 2}=0$ and $c_{-\frac{1}{2}}=1$. Denote   $p_{\frac1 2,\frac1 2}=x$, $p_{-\frac1 2,\frac1 2}=z$, and $p_{-\frac1 2,-\frac1 2}=y$, then the square root of the state has the matrix form
$$
\sqrt{\rho}=
\begin{pmatrix}
y & z\\
\bar{z} &x
\end{pmatrix}.
$$
Here, $x,y\in \mathbb{R}$, and $ z\in \mathbb{C}$.

Take $m=1 /2$ and $n=-1 /2$, Eqs. \eqref{m14} and \eqref{m15} become
\begin{eqnarray*}
y=\frac{s-1}{s+1}x+\frac{u^2}{2(s+1)^2}x+\frac{iuv}{2(1-t)(s+1)}x,\\
y=-\frac{t+1}{1-t}x+\frac{v^2}{2(1-t)^2}x-\frac{iuv}{2(1-t)(s+1)}x.
\end{eqnarray*}
Compare the above two equations, and we have $uv=0.$
Take $m=n=1 /2$, then Eqs. \eqref{m14} and \eqref{m15} become
\begin{eqnarray*}
z=\frac{v}{2(1-t)}x+\frac{iu}{2(s+1)}x,\\
\bar{z}=\frac{v}{2(1-t)}x-\frac{iu}{2(s+1)}x.
\end{eqnarray*}
Take $m=n=-1/ 2$ in Eqs. \eqref{m14} and \eqref{m15}, and there are
\begin{eqnarray*}
(s-1)(\bar{z}-z)=iuy,\\
(t+1)(\bar{z}+z)=vy.
\end{eqnarray*}
Combining the six equations above with the constraint $uv=0$, we find that the parameters $u,v,s,t$ must satisfy one of the following conditions:

\begin{enumerate}[i.]
\item $u=v=0,\quad s+t=0.$\label{en:simple1}
\item $u=0,\quad v\neq0,\quad st=-1.$\label{en:simple2}

\end{enumerate}

If the condition \ref{en:simple1} is satisfied and $s\neq 1$, the matrix can be expressed explicitly as
$$
\sqrt{\rho}=
\begin{pmatrix}
\frac{s-1}{s+1}x & 0\\
0 &x
\end{pmatrix}.
$$
Moreover, from the normalization condition, ${\rm tr} \rho=1$, there is
$$
\rho=
\begin{pmatrix}
\frac{(s^2-1)^2}{2(s+1)^{2}(s^2+1)} & 0\\
0 &\frac{(s+1)^2}{2(s^2+1)}
\end{pmatrix}.
$$
If condition \ref{en:simple2} is satisfied, 
$$
\sqrt{\rho}=x
\begin{pmatrix}
\frac{t+1}{1-t} & \pm \sqrt{\frac{t+1}{1-t}}\\
\pm \sqrt{\frac{t+1}{1-t}} &1
\end{pmatrix},
$$
It can be directly observed that $\rho=\sqrt{\rho}=\ketbra{\psi}{\psi}$is a pure state with $\ket{\psi}=\pm \sqrt{\frac{1+t}{2}}\ket{\frac{1}{2},-\frac{1}{2}}+\sqrt{\frac{1-t}{2}}\ket{\frac{1}{2},\frac{1}{2}}.$

\subsection{General spin systems}

Now we consider the general case where $j\geq 1$. The calculation process analogous to that of the simpler case, where $j=1/2$.

We can simplify Eqs. \eqref{m14} and \eqref{m15} as
\begin{eqnarray}
c_{m-1} p_{m-1,n}&=&\bar{q}p_{m,n}-kc_{m} p_{m+1,n}-lc_{n} p_{m,n+1},\label{j1}\\
c_{n-1} p_{m,n-1}&=&qp_{m,n}-lc_{m} p_{m+1,n}-kc_{n} p_{m,n+1}.\label{j2}
\end{eqnarray}
Here, we define the new parameters 
\begin{eqnarray*}
q&=&\frac{v}{2(1-t)}-\frac{iu}{2(s+1)},\\
k&=&\frac{s+t}{(s+1)(1-t)},\\
l&=&\frac{st+1}{(s+1)(1-t)}.
\end{eqnarray*}

It's reasonable to assume $p_{j,j} \neq0$, otherwise, when $p_{j,j}=0$, then all arbitrary elements become zero, i.e. $p_{m,n}=0,\forall m,n$, which indicates that $ \rho=0 .$  Without loss of generality, we can set $p_{j,j}=1$, and the desired state can be obtained through a normalization process.

The three parameters $q,k,l$ are closely related, as are the originally defined four parameters  $s,t,u,v$. We clarify that these parameters satisfy the proposition outlined below, with detailed proofs provided in the appendix. 
\begin{comment}

 \begin{lemma}\label{lem:1}
    The parameters satisfy $ql=0.$
\end{lemma} 
\begin{lemma}\label{lem:2}
     For a fixed spin number $j(j\geq 1)$, the elements of the $j_{\rm th}$ row in the matrix $\sqrt{\rho}  $ have the following form,
\begin{equation}\label{a1}
p_{j,j-2h}=a_{2h}^{(h)}q^{2h}-a_{2h-2}^{(h)}kq^{2h-2}+...+a_{0}^{(h)}(-k)^h,
\end{equation}
\begin{equation}\label{a2}
p_{j,j-2h-1}=a_{2h+1}^{(h)}q^{2h+1}-a_{2h-1}^{(h)}kq^{2h-1}+...+a_{1}^{(h)}q(-k)^h.
\end{equation}
Here, $h=0,1,\cdots,j-1,j $ and $a_{n}^{(h)}$ is the relevant non-negative coefficient with respect to $h$.
\end{lemma}
\begin{lemma}\label{lem:3}
    If $k=0$, and there is $q=0.$
\end{lemma}
\begin{lemma}\label{lem:4}
    If $k\neq0$ and $q=0$, it can be obtained that $l=0.$
\end{lemma}
    
\end{comment}

\begin{proposition}\label{prop}
    The four parameters $s,t,u,v$ 
    must satisfy one of the conditions:
\begin{enumerate}
    \item $s=-t  $, $v=u=0$.\label{cond1}
    \item $st=-1$, $u=0$. \label{cond2}
\end{enumerate}
\end{proposition}

With this proposition, we can obtain the exact matrix form of minimum uncertainty states for Eq. \eqref{mix3}. When parameters satisfy Condition \ref{cond1}, $u=v=0$ and $s=-t$. Therefore, Eqs. \eqref{m3} and \eqref{m4} become
\begin{eqnarray}
    (s-1)(J_{-}\sqrt{\rho}-\sqrt{\rho}J_{+})+(s+1)(J_{+}\sqrt{\rho}-\sqrt{\rho}J_{-})&=0,\\
    -(s-1)(J_{-}\sqrt{\rho}+\sqrt{\rho}J_{+})+(s+1)(J_{+}\sqrt{\rho}+\sqrt{\rho}J_{-})&=0
\end{eqnarray}
By adding these two equations together, we find that the state $ \rho $ should satisfy
\begin{equation*}
(s+1)J_{+}\sqrt{\rho}=(s-1)\sqrt{\rho}J_{+}.
\end{equation*}
This condition can be expressed in terms of the matrix elements as
\begin{equation*}
(s+1)c_{m-1} p_{m-1,n}=(s-1)c_n p_{m,n+1}.
\end{equation*}
From the iterative relationship, we derive that
\begin{equation*}
p_{m-k,m}=\frac{(s-1)c_{m}}{(s+1)c_{m-k+1}}p_{m-k+1,m+1}=...=0,
\end{equation*}
and
\begin{equation*}
p_{m,m}=s'p_{m+1,m+1}=...=s'^{j-m},
\end{equation*}
where 
\begin{equation}\label{eq:s}
    s'=(s-1)/(s+1).
\end{equation}

As a result, $\rho$ is a diagonal state and the $n_{\rm th}$ element is
\begin{equation}
\rho_{n,n}=s'^{2(j-n)}x^{2}.
\end{equation}
By applying the normalization condition $\textup{tr}\rho=1$, it can be obtained that $  x=\sqrt{\frac{1-s'^2}{1-s'^{4j+2}}}. $

Denote $ \beta=2 \ln |s'| $ and the state can be written as the Gibbs state
$$\rho=\frac{e^{-\beta J_{z}}}{\text{tr}e^{-\beta J_{z}}}.$$

It is important to note that, in the special case $s=1,t=-1,u=v=0$, Eqs. \eqref{m3} and \eqref{m4} reduce to the condition $J_{+}\sqrt{\rho}=0$, which implies $\rho=\ketbra{j,j}{ j,j}.$ This state can be transformed into an arbitrary spin coherent state via displacement operators. This special case illustrates that spin coherent states may be regarded as minimum uncertainty states, aligning with the fact that the coherent states possess the least quantumness \cite{PhysRevA.100.062114}.

When parameters meet Condition \ref{cond2}, Eqs. \eqref{m3} and \eqref{m4} become
\begin{eqnarray}
    (s-1)(J_{-}\sqrt{\rho}-\sqrt{\rho}J_{+})+(s+1)(J_{+}\sqrt{\rho}-\sqrt{\rho}J_{-})&=iu\sqrt{\rho},\\
    (s-1)(J_{-}\sqrt{\rho}+\sqrt{\rho}J_{+})+(s+1)(J_{+}\sqrt{\rho}+\sqrt{\rho}J_{-})&=sv\sqrt{\rho}
\end{eqnarray}
Add these two equations, and there is
\begin{equation}\label{j2}
J_{+}\sqrt{\rho}=\bar{q}\sqrt{\rho}-kJ_{-}\sqrt{\rho}.
\end{equation}
Here, $k=(s-1)/(s+1)$  and $q=(vs)/(2s-2)-(iu)/(2s+2)$.
\begin{lemma}
    The state satisfying Eq. \eqref{j2} is pure.
\end{lemma}
\begin{proof}
    Rewrite Eq. \eqref{j2} in the form
\begin{equation*}
c_{m-1} p_{m-1,n}=\bar{q}p_{m,n}-kc_{m} p_{m+1,n}.
\end{equation*}
To prove the state satisfying Eq. \eqref{j2} is pure, it is only necessary to prove that 
\begin{equation}\label{eq:pur}
    p_{m,n}=\bar{p}_{j,m}p_{j,n}
\end{equation}
Furthermore, the state has the form of $\rho=\sqrt{\rho}=\ketbra{\psi}{\psi},$ and $\ket{\psi}=\sum_{n} p_{j,n}\ket{j,n}.$ 

The Eq. \eqref{eq:pur} can be proven by induction. First, for $m=j-1$, we have
$$p_{j-1,n}=\frac{\bar{q}}{c_{j-1}}p_{j,n}=\bar{p}_{j-1,n}p_{j,n}.$$
Suppose Eq. \eqref{eq:pur} holds for the integer $m$ and we prove that it can also be established for $m-1$. Due to the recurrence relation, there is
\begin{eqnarray*}
p_{m-1,n}&=&\frac{\bar{q}}{c_{m-1}}p_{m,n}-k\frac{c_{m}}{c_{m-1}}p_{m+1,n}\\
&=&(\frac{\bar{q}}{c_{m-1}}p_{m,j}-k\frac{c_{m}}{c_{m-1}}p_{m+1,j})p_{j,n}\\
&=&\bar{p}_{j,m-1}p_{j,n}.
\end{eqnarray*}
Hence, Eq. \eqref{eq:pur} holds and the state is pure.
\end{proof}

Note that Eq. \eqref{j2} is equivalent to
\begin{equation}\label{j3}
\Big(\frac{1}{\sqrt{k}}J_+ +\sqrt{k}J_-\Big)\ket{\varphi}=\frac{\bar{q}}{\sqrt{k}}\ket{\varphi}.
\end{equation}
%\begin{equation}\label{j3}
%\Big((1+k)J_x +i(1-k)J_y \Big) %|\varphi\rangle=\bar{q}|\varphi\rangle.
%\end{equation}

And the solution of \eqref{j3} is 
\begin{equation}\label{eq:solpur}
\ket{\varphi_{n}}=A_{n} e^{-\frac{1}{2} \ln k J_{z}} e^{-\frac{i\pi}{2}J_{y}}\ket{j,n},
\end{equation}
where 
\begin{equation}
    n=\frac{\bar{q}}{2\sqrt{k}}=\frac{vs}{4\sqrt{k}(s-1)}-\frac{iu}{4\sqrt{k}(s+1)},
\end{equation}
and $A_{n}$ is the normalization coefficient \cite{puri1994minimum}. To make the solution sensible, the parameters should fulfill that $n$ is an integer or a half-integer, thus, $u=0$.
%Since $n=\frac{\bar{q}}{2(1+k)}$ is an integer, there is $u=0$. 

To make the proof complete and fluent, here we verify that the state \eqref{eq:solpur} is indeed the solution of Eq. \eqref{j3}. Substitute the state, and we can obtain that 
\begin{equation}
    \begin{aligned}
        &\Big(\frac{1}{\sqrt{k}}J_+ +\sqrt{k}J_-\Big)\ket{\varphi_{n}}\\
        &= \Big(e^{-\frac{1}{2}\ln  k}J_+ +e^{\frac{1}{2}\ln k}J_-\Big)\ket{\varphi_{n}}\\
&= e^{-\frac{1}{2}\ln k J_z}(J_+ + J_-) e^{\frac{1}{2}\ln k J_z}\ket{\varphi_{n}}\\
&=2 e^{-\frac{1}{2}\ln k J_z}J_x e^{\frac{1}{2}\ln k J_z}\ket{\varphi_{n}}\\ 
&=2 e^{-\frac{1}{2}\ln k J_z} e^{-i\frac{\pi}{2}J_y}J_z e^{i\frac{\pi}{2} J_y} e^{\frac{1}{2} \ln k J_z}\ket{\varphi_{n}}\\
&=2A_n   e^{-\frac{1}{2}\ln k J_z} e^{-i\frac{\pi}{2}J_y}J_z \ket{j,n}\\
&=2n A_n \ket{\varphi_{n}}.
    \end{aligned}
\end{equation}

As a conclusion, the minimum uncertainty states must be the rotation of diagonal states or pure states. 
%We replace the two conditions with the terms $s,t,u,v$ :
%\begin{enumerate}
%    \item $s=-t  $, $v=u=0$.\label{cond1}
%    \item $st=-1$, $u=0$. \label{cond2}
%\end{enumerate}
More specifically, when Condition \ref{cond1} holds, the minimum uncertainty state is the diagonal state in the form of 
\begin{equation}\label{eq:mixsol}
    \rho=\sum_{n}\rho_{n,n}R\ket{j,n}\bra{j,n}R^{\dagger},
\end{equation}
where $R$ is an arbitrary rotation operator and
\begin{equation*}
\rho_{n,n}=s'^{2(j-n)}\frac{1-s'^2}{1-s'^{4j+2}},s'\in \mathbb{R}.
\end{equation*}

 When Condition \ref{cond2} is established, the minimum uncertainty state is a pure state of the form  
\begin{equation}\label{eq:minpur}
\ket{\varphi_{m}}=A_{m} R e^{\beta J_{z}} e^{-\frac{i\pi}{2}J_{y}}\ket{j,m},
\end{equation}
where $\beta \in \mathbb{R}$. It should be noted that, in this case, the uncertainty relation \eqref{mix3} reduces to the original Heisenberg uncertainty relation, and Eq. \eqref{eq:minpur} has already been derived and discussed in the literature \cite{puri1994minimum,puri1997coherent}.

Finally, we remark that when $ s=-1$ or $t=1$, the proposition \ref{prop} remains valid after replacing $1-t$ and $s+1$ in the denominator by $t+1$ and $s-1$, respectively, following an entirely analogous analysis. Consequently, the conditions $s=-1$ and $t=1$ hold simultaneously. In this case, the minimum uncertainty state is the pure state $\ket{j,-j}$.

\section{The minimum uncertainty states in the Wick symbol representation}\label{sec:wic}
%\section{Symbol calculus in spin system}
%\cite{Luo1997}
%Symbol calculation in spin system
To determine the minimum uncertainty states associated with the refined uncertainty relation in Eq. \eqref{mix1} for bosonic systems, we specialize to the case where $X$ and $Y$ are identified with the position and momentum operators, respectively, and employ the Wick symbol as the principal analytical tool. Owing to Schwinger’s bosonic representation, a spin system can be mapped to a two-mode bosonic system via
\begin{eqnarray}
    J_+&=&a^{\dagger}b,\\
     J_-&=&ab^{\dagger},\\
     J_z&=&\frac{1}{2}(a^{\dagger}a-b^{\dagger}b).
\end{eqnarray}

Similar to the case for the pseudounitary group SU($1,1$), for which the associated Wick symbol calculus can be constructed from the bosonic Wick symbol \cite{Luo1997}, we develop a Wick symbol calculus for the group SU($2$) relevant to spin systems in this section. Using this framework, we derive the minimum uncertainty states and establish a one-to-one correspondence between the solutions obtained from the two approaches.

\subsection{Symbol calculus in spin systems}
In the Fock-Bargmann representation, every operator $T$ can be represented by the corresponding symbol 
\begin{equation}
    \widehat T(\xi, \eta)=\langle\xi||T||\eta\rangle,
\end{equation}
where $||\xi\rangle=\sum_{m=-j}^j\sqrt{2j\choose {j+m}}\xi^{j+m}\ket{j,m}$ is unnormalized coherent state.  Then the system Hilbert space is
$$\mathcal{H}=\{f: \mathcal{C}\to\mathcal{C}, holomorphic, ||f||^2<\infty\},$$
where the corresponding scalar product is written as 
$$\langle f_1|f_2\rangle=\int \bar f_1(z)f_2(z){\rm d}\mu(z),$$ 
with ${\rm d}\mu(z)=\frac{2j+1}{\pi(1+|z|^2)^{2j+2}}{\rm d}^2z.$
The actions of ladder operators in the Fock-Bargmann representation are
$$\widehat{J_+T}=(2j\bar\xi-\bar\xi^{2}\frac{{\rm d}}{{\rm d}\bar\xi})\widehat T, \qquad \widehat{J_-T}=\frac{{\rm d}}{{\rm d}\bar\xi}\widehat T,$$
$$\widehat{TJ_-}=(2j\eta-{\eta}^2\frac{{\rm d}}{{\rm d}\eta})\widehat T, \qquad \widehat{TJ_+}=\frac{{\rm d}}{{\rm d}\eta}\widehat T.$$

Let $T=\sqrt\rho$ denote the square root of the quantum state. It is important to note that the holomorphic function $\widehat{T}$ must be a polynomial in $\bar\xi$ and $\eta,$ with degree at most $2j$ in each variable. To identify the states that saturate the information-refined uncertainty relation defined in Eq.  Eqs. \eqref{m3} and \eqref{m4} as the differential equations for $\widehat{T}$.

Under the special case $s=-1$ and $t=1$, the Eqs. \eqref{m3} and \eqref{m4} become
\begin{eqnarray*}
\frac{iu}2\widehat T&=&\frac{{\rm d}\widehat T}{{\rm d}\eta}-\frac{{\rm d}\widehat T}{{\rm d}\bar\xi},\label{eq3.0}\\
\frac v2\widehat T&=&\frac{{\rm d}\widehat T}{{\rm d}\bar\xi}+\frac{{\rm d}\widehat T}{{\rm d}\eta},\label{eq4.0}
\end{eqnarray*}
where $u=4{\rm tr}\rho J_y, v=4{\rm tr}\rho J_x,$ of which the solution is 
$$\widehat T(\xi, \eta)=ce^{\frac{v-iu}{4}\bar\xi+\frac{v+iu}{4}\eta}.$$
Thus, the only valid solution (being a polynomial) is 
$$\widehat T(\xi, \eta)=1,$$
which corresponds to the pure state $\rho=\ketbra{j,-j}{ j,-j},$
with $u=v=0.$

%Otherwise, $s\neq-1$ or $t\neq1,$ 

For $s\neq-1$ and $t\neq1$, let
\begin{eqnarray*}
    u'&=&\frac{u}{1+s}, \qquad v'=\frac{v}{1-t},\qquad 
     s'=\frac{s-1}{1+s}, \qquad t'=\frac{1+t}{1-t}.
\end{eqnarray*}
Note that the parameter $s'$ defined here coincides with the symbol introduced in Eq. \eqref{eq:s}. The Eqs. \eqref{m3} and \eqref{m4} become
\begin{eqnarray}
    %\label{eq:diff}
   % \begin{aligned}
        (2j(\bar\xi-\eta)-iu')\widehat T&=&(\bar\xi^2-s')\frac{{\rm d}\widehat T}{{\rm d}\bar\xi}-(\eta^2-s')\frac{{\rm d}\widehat T}{{\rm d}\eta},\label{eq3}\\
(2j(\bar\xi+\eta)-v')\widehat T&=&(\bar\xi^2-t')\frac{{\rm d}\widehat T}{{\rm d}\bar\xi}+(\eta^2-t')\frac{{\rm d}\widehat T}{{\rm d}\eta}.\label{eq4}
   % \end{aligned}
\end{eqnarray}
%where $u'=4{\rm tr}\rho J_y/(1+s), v'=4{\rm tr}\rho J_x/(1-t), s'=(s-1)/(1+s), t'=(1+t)/(1-t).$
%Then we have the valid solutions:
%\noindent If $s\neq-1$ or $t\neq1,$ 
%\noindent 
%The solutions of the differential equations are valid as symbols for the square root of quantum states.

As discussed in Section \ref{sec:mat}, the differential equations \eqref{eq3} and \eqref{eq4} admit valid solutions only when the parameters satisfy certain conditions; in this case, the solutions are holomorphic functions corresponding to the square roots of quantum states in spin systems. We therefore analyze the problem separately for the different parameter regimes.

When Condition \ref{cond1} is satisfied, we have $u=v=0$ and $s=-t$, and the newly defined parameters accordingly satisfy $u'=v'=0, s'=-t'$. In this case, we obtain
\begin{eqnarray}
    %\label{eq:diff}
   % \begin{aligned}
        2j(\bar\xi-\eta)\widehat T&=&(\bar\xi^2-s')\frac{{\rm d}\widehat T}{{\rm d}\bar\xi}-(\eta^2-s')\frac{{\rm d}\widehat T}{{\rm d}\eta},\\
2j(\bar\xi+\eta)\widehat T&=&(\bar\xi^2-s')\frac{{\rm d}\widehat T}{{\rm d}\bar\xi}+(\eta^2-s')\frac{{\rm d}\widehat T}{{\rm d}\eta}.
   % \end{aligned}
\end{eqnarray}
And the solution is $$\widehat T(\xi, \eta)=c(\bar\xi\eta+s')^{2j},$$
which corresponds to the mixed states
$$\rho=\frac{\lambda^{2j}(1-\lambda^2)}{1-\lambda^{4j+2}}\sum_{m=-j}^j \lambda^{2m}\ketbra{j,m}{j,m},$$
where $c=\frac{\lambda^{j}\sqrt{1-\lambda^2}}{\sqrt{1-\lambda^{4j+1}}}$ and $s'=1/\lambda.$
%It corresponds to the mixed states 
%$\rho=\frac{\lambda^j(1-\lambda)}{1-\lambda^{2j+1}}
% \sum_{m=-j}^j \lambda^m|j,m\rangle\langle j,m|,\qquad % \lambda\in[0,1),$$ 
%where $c=\frac{\sqrt{1-\lambda}}{\sqrt{1-\lambda^{2j+1}}}$ and $s'=\sqrt\lambda,$ with
%$$\widehat{\sqrt\rho}=\frac{\sqrt{1-\lambda}}{\sqrt{1-\lambda^{2j+1}}}(\xi\bar\eta+\sqrt\lambda)^{2j}.$$

When Condition \ref{cond2} is satisfied, we have $st=-1 $ and $ u=0$, equivalently, $t'=s'$ and $u'=0$. Under these conditions, the Eqs. \eqref{eq3} and \eqref{eq4} reduce to
\begin{equation}
    \begin{aligned}
        4j{\eta}\widehat T-v'\widehat T&=2({\eta}^2-s')\frac{{\rm d}\widehat T}{{\rm d}{\eta}},\\
        4j\bar\xi \widehat T-v'\widehat T&=2(\bar\xi^2-s')\frac{{\rm d}\widehat T}{{\rm d}\bar\xi}.
    \end{aligned}
\end{equation}
By leveraging the method of separation of variables, we obtain the solution of the differential equations,
\begin{equation}\label{eq39}
    \widehat T=c(\bar\xi+\sqrt{s'})^{j+\frac{v'}{4\sqrt{s'}}}(\bar\xi-\sqrt{s'})^{j-\frac{v'}{4\sqrt{s'}}}({\eta}+\sqrt{s'})^{j+\frac{v'}{4\sqrt{s'}}}({\eta}-\sqrt{s'})^{j-\frac{v'}{4\sqrt{s'}}}
\end{equation}
Moreover, when $v'=0$, the function
$$\widehat T(\xi, \eta)=c(\bar\xi^2-s')^j(\eta^2-s')^j$$
is a polynomial and therefore constitutes a valid solution. In the special case $s'=0$, Eq. \eqref{eq39} reduces to $\widehat T(\xi, \eta)=c\bar{\xi}^{2j} \eta^{2j}$ which corresponds to a spin coherent state $|j,j\rangle.$ When $v'\neq 0$ and $s'\neq 0$, Eq. \eqref{eq39} yields a valid solution provided that $v'/(4\sqrt{s'})$ is an integer or a half-integer, consistent with $j$ being an integer or half-integer.

%When $j\in\mathbb N,$ it corresponds to the states 
%$$\ket{\psi}=\sum_{m=0}^j \frac{{j\choose m}}{\sqrt{2j \choose 2m}} (-s')^{j-m} \ket{j,2m-j}.$$
%When $0\leq j+v'/(4\sqrt{s'})\leq 2j\in\mathbb Z,$ it corresponds to states ...

%If $u', \ v'\neq0, t'=s'\geq0$, we have 
%$$\widehat T(\xi, \eta)=c(\xi\pm\sqrt{s'})^{2j}(\bar\eta\pm\sqrt{s'})^{2j}
%,$$
%and if $u', \ v'\neq0, t'=s'<0$, we have 
%$$\widehat T(\xi, \eta)=c(\xi\pm i\sqrt{-s'})^{2j}(\bar\eta\pm i\sqrt{-s'})^{2j}.$$
%They are special cases of Eq. \eqref{eq39}, and correspond to the normalized spin coherent state, 
%$$\rho=\ketbra{\zeta}{\zeta},$$
%where $\zeta=\frac1{\sqrt{|s'|}} e^{i\frac{k\pi}2},k=0,1,2,3.$
%$$\widehat {T}(\xi, \eta)= \frac{(1+\bar\zeta \xi)^{2j}(1+\zeta\bar\eta)^{2j}}{(1+|\zeta|^2)^{2j}}.$$

%??$\sqrt{s'}=v'/(4j)$

\subsection{One-to-one correspondence}

In the foregoing discussion, we have solved the minimum uncertainty states using two different approaches. In this subsection, we unify the solutions obtained from these two formulations. Since the solutions fall into two distinct classes according to proposition \ref{prop}, we discuss them separately.

When  Condition \ref{cond1} is satisfied,  we have $u'=v'=0, s'=-t'$. The function
$$\widehat T(\xi, \eta)=c(\xi\bar\eta+s')^{2j}$$
with normalization constant $c=\frac{\sqrt{1-s'^2}}{\sqrt{1-s'^{4j+2}}}$ corresponds to the mixed state
$$\rho=\frac{s'^{2j}(1-s'^{2})}{1-s'^{4j+2}}
\sum_{m=-j}^j s'^{2m}\ketbra{j,m}{j,m},\qquad \lambda\in[0,1),$$ which is consistent with the expression given in Eq. \eqref{eq:mixsol}.

When  Condition \ref{cond2} is satisfied, we have $t'=s'$ and $u'=0$.It suffices to show that the states
\begin{equation*}
\ket{\varphi_{n}}=A_{n} e^{-\frac{1}{2} \ln k  J_{z}} e^{-\frac{i\pi}{2}J_{y}}\ket{j,n}
\end{equation*}
have the same expression as Eq. \eqref{eq39} when represented in terms of the Wick symbol.

We therefore compute the Wick-symbol representation of this state. First, we simplify the expression above. Since  $$\widehat T(\xi, \eta)= \langle \eta| \varphi_n\rangle\langle\varphi_n |\xi\rangle,$$ it is sufficient to evaluate only one of the two factors,
\begin{equation}\label{eq:half}
\begin{aligned}
    \langle \varphi_n|\xi\rangle&=\bar{A_n}\bra{j,n}e^{\frac{i\pi}{2}J_{y}}e^{-\frac{1}{2} \ln k J_{z}} \ket{\xi},
\end{aligned}
\end{equation}
and the remaining factor is obtained analogously by complex conjugation.

The unnormalized state has the form
\begin{equation}
    \begin{aligned}
e^{-\frac{1}{2} \ln k J_{z}} \ket{\xi}&=\frac{1}{(1+|\xi|^2)^j}\sum_{m=-j}^{j}e^{-\frac{1}{2} \ln k J_{z}}\sqrt{2j\choose {j+m}}\xi^{j+m}\ket{j,m}\\
&=\frac{k^{\frac{j}{2}}}{(1+|\xi|^2)^j}\sum_{m=-j}^{j}\sqrt{2j\choose {j+m}}\Big( \frac{\xi}{\sqrt{k}}\Big)^{j+m}\ket{j,m}\\
&=\Big(\frac{k+|\xi|^2}{\sqrt{k}+\sqrt{k}|\xi|^2}\Big)^j\ket{\frac{\xi}{\sqrt{k}}}.
    \end{aligned}
\end{equation}
To evaluate the state obtained by the action of the rotation operator $e^{\frac{i\pi}{2}J_{y}}$ acting on the coherent state $\ket{\frac{\xi}{\sqrt{k}}}$, we consider the $2\times 2$ matrix representation of the rotation group, which is a Lie group. In this representation, a frequently used basis of the Lie algebra is given by
\begin{equation}
    J_+=\begin{pmatrix}
        0&1\\
        0&0
    \end{pmatrix},
     J_-=\begin{pmatrix}
        0&0\\
        1&0
    \end{pmatrix},
     J_z=\begin{pmatrix}
        \frac{1}{2}&0\\
        0&- \frac{1}{2}.
    \end{pmatrix}
\end{equation}
Within this framework, the displacement operator takes the form
\begin{equation}
    \Omega(\alpha)
    =\frac{1}{\sqrt{1+|\alpha|^2}}\begin{pmatrix}
             1 & \alpha\\
                -\bar{\alpha}&    1
        \end{pmatrix}, \alpha\in\mathbb{C}.
\end{equation}
as established in this matrix representation \cite{arecchi1972atomic}.

The rotation operators can then be represented as $2\times 2$ matrices as follows:
\begin{equation}
    \begin{aligned}
        e^{\frac{i\pi}{2}J_{y}}&=   e^{\frac{\pi}{4}(J_+-J_-)}\\
        &=\Omega(\frac{\pi}{4})\\
        &=\begin{pmatrix}
             \frac{1}{\sqrt{2}} &  \frac{1}{\sqrt{2}}\\
                -\frac{1}{\sqrt{2}}&    \frac{1}{\sqrt{2}}
        \end{pmatrix},\\
    \Omega(\frac{\xi}{\sqrt{k}})   &=\Omega(\tau)\\
    &=\frac{1}{\sqrt{1+|\tau|^2}}\begin{pmatrix}
             1 & \tau\\
                -\bar{\tau}&    1
        \end{pmatrix},
    \end{aligned}
\end{equation}
where $\tau=\xi/\sqrt{k}$.

The product of the two rotation operators above is itself a rotation, which can be decomposed into a rotation about the $ z$-axis and a displacement operator, denoted by $e^{-i\Psi J_z}$ and $\Omega(T)$, respectively. Explicitly, 
\begin{equation}
    \begin{aligned}
        \Omega(\frac{\pi}{4})\Omega(\tau)&=\frac{1}{\sqrt{2(1+|\tau|^2)}}\begin{pmatrix}
             1 -\bar{\tau}& 1+\tau\\
               -1 -\bar{\tau}&    1-\tau
        \end{pmatrix}\\
        &=    \Omega(T)e^{-i\Psi J_z}\\
        &=\frac{1}{\sqrt{1+|T|^2}}\begin{pmatrix}
            e^{-i\Psi/2}& Te^{i\Psi/2}\\
               -\bar{T}e^{-i\Psi/2}&    e^{i\Psi/2}
        \end{pmatrix}.
    \end{aligned}
\end{equation}
 %  e^{-\frac{1}{2} \ln k J_{z}}        &=\begin{pmatrix}
             %e^{-\frac{ \ln k}{4} } &  0\\
                %0&   e^{ \frac{\ln k}{4}}
       % \end{pmatrix}\\
Therefore, $T=(1+\tau)/(1-\tau)=(\sqrt{k}+\xi)/(\sqrt{k}-\xi)$ and $e^{i\Psi/2}=(1-\tau)/|1-\tau|=(\sqrt{k}-\xi)/|\sqrt{k}-\xi|$. And Eq.~\eqref{eq:half} becomes
\begin{equation}
    \begin{aligned}
        & \bar{A}_n\Big(\frac{k+|\xi|^2}{\sqrt{k}+\sqrt{k}|\xi|^2}\Big)^j\bra{j,n}\Omega(T)e^{-i\Psi J_z}    \ket{j,-j}\\
        &=\bar{A}_n e^{i j\Psi }\sqrt{2j\choose {j+n}}\Big(\frac{k+|\xi|^2}{\sqrt{k}+\sqrt{k}|\xi|^2}\Big)^j T^{j+n}\\
        &=\bar{A}_n \sqrt{2j\choose {j+n}}\Big(\frac{k+|\xi|^2}{\sqrt{k}+\sqrt{k}|\xi|^2}\Big)^j (\sqrt{k}+\xi)^{j+n}(\sqrt{k}-\xi)^{-(j+n)}(\sqrt{k}-\xi)^{2j}|\sqrt{k}-\xi|^{2j}\\
        &\bar{A}_n (-1)^{j-n}\sqrt{2j\choose {j+n}}\Big(\frac{k+|\xi|^2}{(1+|\xi|^2)|k-\sqrt{k}\xi|^2}\Big)^j (\sqrt{k}+\xi)^{j+n}(\xi-\sqrt{k})^{j-n}
    \end{aligned}
\end{equation}
Since $st=-1,$ we have $k=s'=(s-1)/(s+1)$ and $v'/(4\sqrt{s'})=(vs)/(4\sqrt{(s^2-1)})=n $. Hence, this expression is consistent with that in Eq.\eqref{eq39}.

\subsection {Comparison with bosonic case}

In Holstein-Primakoff representation, the angular momentum operators can be described by a one-mode bosonic field \cite{holstein1940field}, which is different from Schwinger's representation. Explicitly,
\begin{eqnarray*}
    J_{+}&=&a^{\dagger}\sqrt{2j-a^{\dagger }a},\\
    J_{-}&=&a^{\dagger}\sqrt{2j-a^{\dagger }a}a,\\
    J_{z}&=&a^{\dagger}a-j.\\
\end{eqnarray*}
In the limit case $j\rightarrow \infty$, the ladder operators and $J_z$ approaches to bosonic operators \cite{arecchi1972atomic},
\begin{eqnarray*}
    \lim_{j\rightarrow \infty}\frac{J_+}{\sqrt{2j}}&=&a^{\dagger},\\
      \lim_{j\rightarrow \infty}\frac{J_-}{\sqrt{2j}}&=&a,\\
        \lim_{j\rightarrow \infty}J_z +j&=&a^{\dagger}a.
\end{eqnarray*}
In the limit $j\to\infty$, Dicke states and spin coherent states converge to Fock states and bosonic coherent states, respectively. To illustrate this correspondence, we focus on the case of coherent states and consider the following limit,
\begin{equation}
    \begin{aligned}
    \ket{z} &= \lim_{j\rightarrow \infty}  \ket{\zeta}=    (1+|\zeta|^2)^{-j}e^{\zeta J_+}\ket{j,-j}\\
      &=\lim_{j\rightarrow \infty}(1-|\zeta|^2)^{-j}e^{\zeta\sqrt{2j}a^\dagger}\ket{0}.
    \end{aligned}
\end{equation}
Since the bosonic coherent state can be written as $\ket{z}=e^{-|z|^2/2}e^{z a^\dagger}\ket{0}$, the two parametrizations are related by $z=\sqrt{2j}\zeta$.
%with $z=\sqrt{2j}\zeta$ and the limit $j\to\infty$, we have $|\zeta\rangle\to e^{|z|^2/2}|z\rangle,$ which are unnormalized single-mode coherent states.
When  Condition \ref{cond1} is satisfied, under the limit $j\to\infty,$ the function
\begin{equation*}
    \widehat T(\xi, \eta)=c(\xi\bar\eta+s')^{2j}\to e^{\alpha\bar\beta/s'}
\end{equation*}
where $ \alpha=\sqrt{2j}\xi, \beta=\sqrt{2j}\eta$.
Hence, the limit is a Gaussian thermal state.

When  Condition \eqref{cond2} is satisfied, under the limit $j\to\infty$, the function 
\begin{equation*}
\begin{aligned}
     \widehat T(\xi, \eta)&=c(\xi+\sqrt{s'})^{j+\frac{v'}{4\sqrt{s'}}}(\xi-\sqrt{s'})^{j-\frac{v'}{4\sqrt{s'}}}(\bar{\eta}+\sqrt{s'})^{j+\frac{v'}{4\sqrt{s'}}}(\bar{\eta}-\sqrt{s'})^{j-\frac{v'}{4\sqrt{s'}}}\\
   &  \rightarrow     c'e^{\frac{\alpha^2+\bar\beta^2}{2s'}}   \end{aligned}
\end{equation*}
corresponds to a squeezed coherent state, which is pure Gaussian state.

By incorporating arbitrary rotation operations into the spin minimum uncertainty states—operations that correspond to displacement operations in the bosonic setting—we obtain two distinct classes of states: mixed states and pure states. In the mixed-state sector, the spin minimum uncertainty states converge, in the limit $j\rightarrow \infty$, to displaced thermal states. In contrast, the pure spin minimum uncertainty states correspond, in this limit, to pure Gaussian states.

It is worth noting that, in the bosonic context, minimum uncertainty states associated with the uncertainty relation 
\begin{equation}\label{eq:bosonic}
    U(\rho,Q)U(\rho,P)\geq \frac{1}{4},
\end{equation}
where $Q$ and $P$denote the position and momentum operators, respectively, have been fully characterized in the literature \cite{Fu2020}. It has been shown that Gaussian states exhaust the class of minimum uncertainty states for the refined uncertainty relation in bosonic systems. 

In the large-spin limit, one may intuitively expect the spin minimum uncertainty states to converge to their bosonic counterparts, since the spin uncertainty relation Eq. \eqref{mix3} bosonic uncertainty relation converges to the bosonic uncertainty relation Eq. \eqref{eq:bosonic}. However, when taking this limit at the level of the spin solutions, only pure Gaussian states and displaced thermal states are recovered, while general mixed Gaussian states are absent. This discrepancy originates from the boundary conditions intrinsic to the spin system. More concretely, Schwinger’s bosonic representation reveals that a spin system corresponds to a two-mode bosonic system with a fixed total excitation number. This constraint enforces conservation of the total photon number and has no analogue in the unconstrained bosonic minimum uncertainty states. As a result, general mixed Gaussian states, which require fluctuations in the total excitation number, are excluded in the spin setting, whereas no such restriction arises when solving the minimum uncertainty problem directly in the bosonic system.

%|\zeta\rangle=\sum_{m=-j}^j\sqrt{2j\choose {j+m}}\zeta^{j-m}|j,m\rangle,$ 
%\noindent When $u'=v'=0, t'=s'$, 
 %$$\widehat T(\xi, \eta)=c(\xi^2-s')^j(\bar\eta^2-s')^j\to e^{(\alpha^2+\bar\beta^2)/(2s')}/s'^2,$$
 
%\noindent When $u'=v'=0, t'=-s'$,
%$$\widehat T(\xi, \eta)=c(\xi\bar\eta+s')^{2j}\to e^{\alpha\bar\beta/s'}.$$

%\noindent When $u',v'\neq0, t'=s'$, 
%$$\widehat {T}(\xi, \eta)= \frac{(1+\bar\zeta \xi)^{2j}(1+\zeta\bar\eta)^{2j}}{(1+|\zeta|^2)^{2j}}\to e^{\bar\beta z+\alpha \bar z-|z|^2},$$
%where $z=\sqrt{2j}\zeta.$

\section{Summary}\label{sec:sum}

In this work, we investigate the minimum uncertainty states associated with the information-refined uncertainty inequality in general spin systems. We derive explicit expressions for these minimum uncertainty states using two different methods and establish a one-to-one correspondence between the resulting solutions. To make a comparison with the bosonic case, we analyze the large-spin limit $j\rightarrow \infty$, and show that the spin minimum uncertainty states converge to a specific subclass of Gaussian states, which are themselves minimum uncertainty states under the bosonic case. For future work, we will explore minimum uncertainty states within a more general theoretical framework and investigate their potential applications and physical implications.

\vskip 0.5cm \noindent {\bf Acknowledgements}.
 This work was supported by the National Natural Science  Foundation of China, Grant Nos. 12401609 and 12471433, and the Youth Innovation Promotion Association of CAS, Grant No. 2023004, and the Beijing Natural Science Foundation, Grant No. Z250004.

\appendix
\section{proof of Proposition 1}
    \begin{lemma}\label{lem:1}
    The parameters satisfy $ql=0.$
\end{lemma} 
\begin{proof}
 In Eqs. \eqref{j1} and \eqref{j2}, let $m=j$ and $ n=j$, then
\begin{eqnarray*}
p_{j-1,j}&=&\frac{\bar{q}}{c_{j-1}},\\
p_{j,j-1}&=&\frac{q}{c_{j-1}}.
\end{eqnarray*}
In Eqs. \eqref{j1} and \eqref{j2}, let $m=j$ and $n=j-1$, then
\begin{eqnarray*}
p_{j-1,j-1}&=&\frac{\bar{q}q}{c_{j-1}^{2}}-l,\\
p_{j,j-2}&=&\frac{q^2}{c_{j-1}c_{j-2}}-\frac{kc_{j-1}}{c_{j-2}}.
\end{eqnarray*}
If $m=j$ and $n=j-2$ in Eqs. \eqref{j1} and \eqref{j2}, it can be obtained 
\begin{equation*}
p_{j-1,j-2}=\frac{\bar{q}q^2}{c_{j-1}^{2}c_{j-2}}-\frac{\bar{q}k}{c_{j-2}}-\frac{lqc_{j-2}}{c_{j-1}^{2}}.
\end{equation*}
Finally, let $m=j-1$ and $n=j-1$ in Eqs. \eqref{j1} and \eqref{j2}, and there is
\begin{equation*}
p_{j-1,j-2}=\frac{\bar{q}q^2}{c_{j-1}^{2}c_{j-2}}-\frac{\bar{q}k}{c_{j-2}}-\frac{2lq}{c_{j-2}}.
\end{equation*}
Compare the two expressions of $p_{j-1,j-2}$ in the two above equations and it can be obtained that
\begin{equation*}
\frac{2lq}{c_{j-2}}=\frac{lqc_{j-2}}{c_{j-1}^{2}},
\end{equation*}
since $2c_{j-1}^2=4j$ and $c_{j-2}^2=4j-2$,  we have proved that $ql=0.$
\end{proof}

\begin{lemma}\label{lem:2}
     For a fixed spin number $j(j\geq 1)$, the elements of the $j_{\rm th}$ row in the matrix $\sqrt{\rho}  $ have the following form,
\begin{equation}\label{a1}
p_{j,j-2h}=a_{2h}^{(h)}q^{2h}-a_{2h-2}^{(h)}kq^{2h-2}+...+a_{0}^{(h)}(-k)^h,
\end{equation}
\begin{equation} \label{a2}
p_{j,j-2h-1}=a_{2h+1}^{(h)}q^{2h+1}-a_{2h-1}^{(h)}kq^{2h-1}+...+a_{1}^{(h)}q(-k)^h.
\end{equation}
Here, $h=0,1,\cdots,j-1,j $ and $a_{n}^{(h)}$ is the relevant non-negative coefficient with respect to $h$.
\end{lemma}
\begin{proof}
    We prove the lemma by induction. First, when $h=0$,
\begin{eqnarray*}
p_{j,j}&=&1,\\
p_{j,j-1}&=&\frac{q}{c_{j-1}}.
\end{eqnarray*}
Eqs. \eqref{a1} and \eqref{a2} hold. Suppose that for a positive integer $h$, Eqs. \eqref{a1} and \eqref{a2} hold, it can be directly verified that these two equations can be established for $h+1$.
\end{proof}
\begin{lemma}\label{lem:3}
    If $k=0$, and there is $q=0.$
\end{lemma}

\begin{proof}
    According to Lemma \eqref{lem:2}, if $k=0$, we have
\begin{equation*}
p_{j,-j}=a_{2j}^{(j)}q^{2j}.
\end{equation*}
Take $m=j$  and $n=-j$ in Eqs. \eqref{j1} and \eqref{j2}, then
\begin{equation*}
0=a_{2j}^{(j)}q^{2j+1},
\end{equation*}
which implies that $q=0.$
\end{proof}
\begin{lemma}\label{lem:4}
    If $k\neq0$ and $q=0$, it can be obtained that $l=0.$
\end{lemma}
\begin{proof}
    We first prove when $k\neq0,q=0$, $j$ is an integer. Otherwise, suppose $j$ is a half-integer, then $p_{j,-j+1}\propto (-k)^{\frac{2j-1}{2}}$. Let $m=j$  and $n=-j$ in Eq. \eqref{j1}, then $k=0$, which is a contradiction.

Next, from observation in proof of Lemma \eqref{lem:1}, when $q=0$ and $j$ is an integer, the elements satisfy 
\begin{eqnarray*}
p_{j,-j}&=&a_{0}^{(j)}(-k)^{j},\\
p_{j,-j+2}&=&a_{0}^{(j-1)}(-k)^{j-1}.
\end{eqnarray*}
In Eqs. \eqref{j1} and \eqref{j2}, let $m=j$ and  $n=-j+1$, then 
\begin{equation*}
p_{j-1,-j+1}=(-1)^{j}\frac{c_{-j+1}}{c_{j-1}}a_{0}^{(j-1)}lk^{j-1},
\end{equation*}
In Eqs. \eqref{j1} and \eqref{j2}, let $m=j-1$ and $n=-j$, then
\begin{eqnarray*}
0&=&-lc_{j-1}p_{j,-j}-kc_{-j}p_{j-1,-j+1}\\
&=&(-1)^{j+1}c_{j-1}a_{0}^{(j)}lk^{j}+(-1)^{j+1}c_{j-2}a_{0}^{(j-1)}lk^{j}.
\end{eqnarray*}
Thus, $l=0.$

\end{proof}
Combining lemmas \ref{lem:1}, \ref{lem:3}, and \ref{lem:4}, it can be directly observed that the three parameters must satisfy one of the two conditions:

\begin{enumerate}
    \item $k=0  $, $q=0$, and there is no constraint on $l$.
    \item $k\neq 0  $, $l=0$, and there is no constraint on $q$.
\end{enumerate}
We replace the two conditions with the terms $s,t,u,v$ , then obtain the two conditions in Proposition 1:
\begin{enumerate}
    \item $s=-t  $, $v=u=0$.%\label{cond1}
    \item $st=-1$, $u=0$. %\label{cond2}
\end{enumerate}

%%%%%%%%%%%%%%%%%%%%%%%%%%%%%%%%%%%%%%%%
% choose a style
%\bibliographystyle{ieeetr}
%\bibliographystyle{unsrt}
\bibliographystyle{apsrev4-2}
\bibliographystyle{apsrev4-1}
\bibliography{bibsmus}
%%%%%%%%%%%%%%%%%%%%%%%%%%%%%%%%%%%%%%%%
\end{document}